\documentclass[letterpaper, 10 pt, conference]{ieeeconf}
\IEEEoverridecommandlockouts
\usepackage{amssymb} %
\usepackage{amsmath} %
\usepackage{amsthm} %
\usepackage{mathrsfs} %
\usepackage{mathtools} %
\usepackage{cite}
\usepackage{multirow}
\usepackage{lipsum}
\usepackage{textcomp} %
\usepackage{subfig} %
\usepackage{url}
\usepackage{soul}

\newtheorem{thm}{Theorem}
\newtheorem{lem}{Lemma}
\newtheorem{defn}{Definition}

\newtheorem{prop}{Proposition}
\newtheorem{prob}{Problem}

\newcommand{\Prob}{\mathbb{P}}
\newcommand{\Exp}{\mathbb{E}}

\newcommand{\deff}{\stackrel{\triangle}{=}}

\usepackage{xspace}
\usepackage{amsmath,amssymb,amsfonts}

\newcommand{\bs}[1]{\ensuremath{\boldsymbol{#1}}}

\newcommand{\bp}{\ensuremath{\bs p}\xspace}
\newcommand{\bq}{\ensuremath{\bs q}\xspace}

\newcommand{\bw}{\ensuremath{\bs w}\xspace}
\newcommand{\bx}{\ensuremath{\bs x}\xspace}
\newcommand{\by}{\ensuremath{\bs y}\xspace}
\newcommand{\bz}{\ensuremath{\bs z}\xspace}

\newcommand{\bW}{\ensuremath{\bs W}\xspace}
\newcommand{\bX}{\ensuremath{\bs X}\xspace}
\newcommand{\bY}{\ensuremath{\bs Y}\xspace}

\newcommand{\pci}[1]{\ensuremath{\bs{p}_{\boldsymbol{c}_{i}}}\xspace}

\def\QED{\ensuremath{\rule[0pt]{1.5ex}{1.5ex}}\xspace}

\usepackage{environ}
\NewEnviron{killcontents}{}

\begin{document}
\title{Scalable Underapproximation for the Stochastic Reach-Avoid Problem for
High-Dimensional LTI Systems using Fourier Transforms}
\author{Abraham P. Vinod and Meeko M. K. Oishi
\thanks{This material is based upon work supported by the National Science Foundation under Grant Number IIS-1528047, CMMI-1254990 (Oishi, CAREER), and CNS-1329878. Any opinions, findings, and conclusions or recommendations expressed in this material are those of the authors and do not necessarily reflect the views of the National Science Foundation. 
    \newline
    Abraham Vinod and Meeko Oishi are with Electrical and Computer Engineering,
University of New Mexico, Albuquerque, NM; e-mail: aby.vinod@gmail.com,
oishi@unm.edu (corresponding author)}
}
\date{}
\maketitle

\begin{abstract} 

We present a scalable underapproximation of the terminal hitting time stochastic
reach-avoid probability at a given initial condition, for verification of
high-dimensional stochastic LTI systems.  
While several approximation techniques have been proposed to alleviate the curse of dimensionality associated with dynamic programming, these techniques are limited and cannot handle larger, more realistic systems.  We present a scalable method that uses Fourier transforms to compute an underapproximation of the reach-avoid probability for systems with disturbances with arbitrary probability densities.  We characterize sufficient conditions for Borel-measurability of the value functions. We exploit fixed control sequences parameterized by the initial condition (an open-loop control policy) to generate the underapproximation. For Gaussian disturbances, the underapproximation can be obtained using existing efficient algorithms by solving a convex optimization problem. Our approach produces non-trivial lower bounds and is demonstrated on a chain of integrators with 40 states.

\end{abstract}


\section{Introduction}
\label{sec:introduction}

Reachability analysis of discrete-time stochastic dynamical systems is an
established verification tool that provides probabilistic guarantees of safety
or performance, and has been applied to problems in
fishery management and mathematical finance~\cite{SummersAutomatica2010}, motion
planning in robotics~\cite{VinodHSCC2017,HomChaudhuriACC2017,MaloneHSCC2014},
spacecraft docking~\cite{LesserCDC2013}, and autonomous
survelliance~\cite{kariotoglou2015multi}.  In~\cite{SummersAutomatica2010}, two
classes of problems characterize verification over a finite
horizon --- first hitting time and terminal hitting time -- and dynamic programming approaches are formulated to solve both (similarly to~\cite{AbateAutomatica2008,AbateHSCC2007}). %
We focus on the finite time
horizon \emph{terminal} hitting time stochastic reach-avoid problem (referred to
here as the {\em terminal time problem} for convenience), that is, computing
the probability of hitting a target set at the terminal time, while avoiding an
unsafe set during all the preceding time steps.
Specifically, we construct an underapproximation to the terminal time problem 
from a known initial point, as opposed to 
the typical stochastic reach-avoid problem.  This could be used
as a query, for example, in evaluating feasibility of an initial trajectory 
an optimization problem. 

The dynamic programming-based discretization approach (DPBDA), proposed
in~\cite{AbateHSCC2007}, approximately computes value
functions for the terminal time problem, but relies on
gridding, and hence suffers from the well-known curse of dimensionality.  
Attempts to circumvent this problem, via approximate dynamic programming
\cite{KariotoglouECC2013,KariotoglouSCL2016, ManganiniCYB2015}, Gaussian
mixtures \cite{KariotoglouSCL2016}, particle filters \cite{ManganiniCYB2015,
LesserCDC2013}, and convex chance-constrained optimization \cite{LesserCDC2013,
kariotoglou2015multi}, have been applied to systems that are at
most 10-dimensional -- far beyond the scope of what is
possible with DPBDA, but not scalable to larger problems.

In this paper, %
we first characterize
sufficient conditions for Borel-measurability of the value functions 
for the terminal time problem (characterized so far only for the first 
hitting time problem \cite{DingAutomatica2013}).
Using
conditional expectations, we then establish that 
an open-loop formulation 
provides an underapproximation of the stochastic reach-avoid probability
for linear systems \cite{LesserCDC2013}.
We propose a scalable
Fourier transform-based underapproximation (FTBU), for the terminal time problem, exploiting our prior work on uncontrolled stochastic reachable sets \cite{VinodHSCC2017}. For
an arbitrary probability density, the FTBU
solves an optimization problem with a multi-dimensional integration as the
objective function. 
For Gaussian disturbances, 
the objective function can be computed 
efficiently via existing algorithms \cite{GenzJCGS1992}, and
the optimization problem is log-concave.
Our approach does not require gridding of state, input, or disturbance spaces, and has
low memory requirements in contrast to DPBDA. 
Our main contribution is twofold: 1) a Fourier transform-based
underapproximation of the terminal hitting time stochastic reach-avoid
probability from a known initial condition, based on open-loop control
sequences, and 2) the underlying theory that enables us to exploit measurability
and convexity properties to assure a computationally feasible approach.  We
extend our previous work on Fourier transform-based stochastic reachable sets
for uncontrolled systems \cite{VinodHSCC2017} to systems with control inputs,
although here we do not seek to compute the stochastic reach-avoid
set~\cite{GleasonCDC2017}.

In Section~\ref{sec:preliminaries}, we describe the terminal time problem, its open-loop
approximation, and relevant properties from probability theory and Fourier
analysis.  Section~\ref{sec:theorRes}
presents sufficient conditions for Borel-measurability, and establishes the
underapproximation result linking the problems 
in~\cite{SummersAutomatica2010} and~\cite{LesserCDC2013}.
Section~\ref{sec:solveProbB} presents the FTBU and specialized results for Gaussian disturbances. 
We demonstrate scalability in Section~\ref{sec:numex}, through application to a 40D chain of integrators. 
Section~\ref{sec:conc} concludes the work.

\section{Preliminaries and Problem Formulation}
\label{sec:preliminaries}

\label{sub:notation}

We denote the Borel $\sigma$-algebra by $ \mathscr{B}(\cdot)$, a discrete-time
time interval by $ \mathbb{N}_{[a,b]}$ for $a,b\in \mathbb{N}$ and $a\leq b$, 
which inclusively enumerates all integers in between $a$ and $b$,
random vectors with
bold case, and non-random vectors with an overline. The indicator function of a
non-empty set $ \mathcal{S}$ is denoted by $1_{\mathcal{S}}(\bar{y})$, such that
$1_{\mathcal{S}}(\bar{y})=1$ if $\bar{y}\in \mathcal{S}$ and is zero otherwise.
We denote the
$p$-dimensional identity matrix by $I_p$, and the matrix with all
entries as ones by $\bar{1}_{p\times q}\in \mathbb{R}^{p\times q}$.

\subsection{Probability theory}

A random vector $\by$ is a measurable transformation defined in the probability
space $(\Omega,\mathscr{Y}, \Prob)$ with sample space $\Omega$, $\sigma$-algebra
$ \mathscr{Y}$, and probability measure over $ \mathscr{Y}$, $\Prob$.  A
sub-$\sigma$-algebra of $ \mathscr{Y}$ is a $\sigma$-algebra whose members also
belong to $ \mathscr{Y}$.
The \emph{minimal} $\sigma$-algebra of $\by$, the smallest sub-$\sigma$-algebra
of $ \mathscr{Y}$ over which $\by$ is measurable, is denoted by
$\sigma(\by)\subset \mathscr{Y}$.  We typically consider Borel-measurable random
vectors, $\by: \mathbb{R}^p \rightarrow \mathbb{R}^p$ with $\Omega=
\mathbb{R}^p$ and $ \mathscr{Y}=\sigma(\by)= \mathscr{B}( \mathbb{R}^p)$. 
For $N\in \mathbb{N}$, a random process is a sequence of random vectors
${\{\by_k\}}_{k=0}^N$ where the random vectors $\by_k$ are defined in the
probability space $(\Omega,\mathscr{Y}, \Prob)$. The random vector $\bY={[\by_0\
\by_1\ \ldots\ \by_N]}^\top$ is defined in the probability space $(\Omega^{N+1},
\sigma(\bigtimes_{k=0}^N \mathscr{Y}_k), \Prob_{\bY})$, with $\Prob_{\bY}$
induced from $\Prob$.  See~\cite{GubnerProbability2006,ChowProbability1997} for
details.

Conditional expectations transforms random variable $\bp$ in
$(\Omega,\mathscr{P}, \Prob)$ whose mean exists ($ \vert \Exp\left[
\bp\right]\vert \leq \infty$) to a sub-$\sigma$-algebra $ \mathscr{Q}$, i.e, $
\mathscr{Q}$ is a $\sigma$-algebra with all its members containing in $
\mathscr{P}$. For a $ \mathscr{Q}$-measurable random variable $\bq$ such that
$\int_{ \mathcal{S}} \bp d\Prob=\int_{ \mathcal{S}} \bq d\Prob$ for all $
\mathcal{S}\in \mathscr{Q}$, $\bq= \Exp\left[ \bp \middle\vert \mathscr{Q}
\right]$ almost surely (a.s.)~\cite[Sec. 7.1, Thm. 1]{ChowProbability1997}. 
\begin{enumerate}
    \item[P1)] For $\mathscr{Y}$-measurable random variables $\by_1,\by_2$ with
        finite means, a sub-$\sigma$-algebra $ \mathscr{G} \subset \mathscr{Y}$,
        if $ \by_1\leq\by_2$ a.s., then $\Exp\left[ \by_1  \middle\vert
        \mathscr{G} \right] \leq \Exp\left[ \by_2 \middle\vert \mathscr{G}
        \right]$ a.s.  ~\cite[Sec.  7.1, eq. (14, i and
        iii)]{ChowProbability1997}.
    \item[P2)] For bounded random variables $\by,\bz$ that are $
        \mathscr{Y},\mathscr{Z}$-measurable respectively, if $\vert \Exp\left[
        \by\bz \right] \vert \leq\infty$, then $ \by\Exp\left[ \bz
        \middle\vert \mathscr{Y} \right]=\Exp\left[ \by\bz \middle\vert
        \mathscr{Y} \right]$ a.s.~\cite[Sec. 7.1, Thm.  3]{ChowProbability1997}.
    \item[P3)] $ \Exp\left[ \Exp\left[ \by \middle\vert \mathscr{G}_2 \right]
        \middle\vert \mathscr{G}_1 \right]= \Exp\left[ \by \middle\vert
        \mathscr{G}_1 \right]$ a.s. if $ \mathscr{G}_1\subset \mathscr{G}_2
        \subset \sigma(\by)$ and $ \Exp\left[  \by \right]\leq
        \infty$~\cite[Sec.  7.1, eq.  (14, v)]{ChowProbability1997}.
\end{enumerate}

The characteristic function (CF) of a random vector $\by\in \mathbb{R}^{p}$ with
probability density function (PDF) $\psi_{\by}(\bar{z})$ is
\begin{align}
    \Psi_{\by}(\bar{\alpha})&\triangleq
    \Exp_{\by}\left[\mathrm{exp}\left({j\bar{\alpha}^\top\by}\right)\right] \nonumber \\
                       &=\int_{\mathbb{R}^p}e^{j\bar{\alpha}^\top\bar{z}}
    \psi_{\by}(\bar{z})d\bar{z}=
    \mathscr{F}\left\{\psi_{\by}(\cdot)\right\}(-\bar{\alpha})\label{eq:cfun_def}
\end{align}
where $ \mathscr{F}\{\cdot\}$ denotes the Fourier transformation operator and
$\bar{\alpha}\in \mathbb{R}^{p}$. Given a CF $\Psi_{\by}(\bar{\alpha})$, the
PDF can be computed as
\begin{align}
    \psi_{\by}(\bar{z})&=
    \mathscr{F}^{-1}\left\{\Psi_{\by}(\cdot)\right\}(-\bar{z})
    \nonumber \\
    &={\left(\frac{1}{2\pi}\right)}^p\int_{
\mathbb{R}^p}e^{-j\bar{\alpha}^\top\bar{z}}
\Psi_{\by}(\bar{\alpha})d\bar{\alpha}\label{eq:cfun_ift}
\end{align}
where $ \mathscr{F}^{-1}\{\cdot\}$ denotes the inverse Fourier transformation
operator and $d\bar{\alpha}$ is short for $d\alpha_1d\alpha_2\ldots d\alpha_p$.
Since PDFs are absolutely integrable, every PDF has a unique CF. See~\cite[Sec.
1]{SteinFourier1971},~\cite[Sec. 22.6]{CramerMathematical1961},~\cite[Sec.
7.2, 8.2]{GubnerProbability2006}, \cite[Sec. 2.1]{VinodHSCC2017} for more
details about CFs.

\subsection{Terminal stochastic reach-avoid analysis}
\label{sub:stochRA}

Consider the discrete-time stochastic LTI system,
\begin{align}
    \bx_{k+1}&=A\bx_k+B\bar{u}_k+\bw_k\label{eq:sys}
\end{align}
with state $\bx_k\in \mathcal{X}=\mathbb{R}^n$, input $\bar{u}_k\in
\mathcal{U}\subseteq \mathbb{R}^m$, disturbance $\bw_k\in \mathcal{W}\subseteq
\mathbb{R}^n$, and matrices $A,B$ assumed to be of appropriate dimensions.
We assume that $ \mathcal{U}$ is compact,
$\bw_k$ is absolutely continuous with a known PDF $\psi_{\bw}$, and the random
process $\bw[\cdot]$ is independent and identical distributed (IID).
Let $N$ be a finite time horizon. For any given sequence of (non-random) inputs
$\bar{u}[\cdot]$ and an initial condition $\bar{x}_0\in \mathcal{X}$, the state
$\bx_k$ is a random vector for all $k\in \mathbb{N}_{[1,N]}$ via \eqref{eq:sys}.  

The system \eqref{eq:sys} can be equivalently described by a Markov control
process 
with stochastic kernel that is a 
Borel-measurable function $Q: \mathscr{B}( \mathcal{X}) \times \mathcal{X}
\times \mathcal{U} \rightarrow [0,1]$, which assigns to each $\bar{x} \in \mathcal{X}$
and $\bar{u}\in \mathcal{U}$ a probability measure on the Borel space
$(\mathcal{X}, \mathscr{B}( \mathcal{X}))$. For \eqref{eq:sys},
\begin{align}
    Q(d\bar{y}\vert
    \bar{x},\bar{u})&=\psi_{\bw}(\bar{y}-A\bar{x}-B\bar{u})d\bar{y}.\label{eq:stoch_kernel}
\end{align}
We define a \emph{Markov policy} $\pi=(\mu_0,\mu_1,\ldots,\mu_{N-1})\in
\mathcal{M}$ as a sequence of universally measurable
maps $\mu[\cdot]:
\mathcal{X} \rightarrow \mathcal{U}$. 
The random vector $\bX=
[\bx^\top_1\ \bx^\top_2\ \ldots\ \bx^\top_N]^\top$, defined in 
$( \mathcal{X}^N, \mathscr{B}( \mathcal{X}^N),
\Prob_{\bX}^{\bar{x}_0,\pi})$~\cite{SummersAutomatica2010}, has 
probability measure $\Prob_{\bX}^{\bar{x}_0,\pi}$ defined using
$Q$~\cite[Prop. 7.45]{BertsekasSOC1978}.

Let $ \mathcal{S}, \mathcal{T} \in \mathscr{B}( \mathcal{X})$. 
Define the \emph{terminal time probability},
$\hat{r}_{\bar{x}_0}^\pi( \mathcal{S}, \mathcal{T})$, 
for known $\bar{x}_0$ and $\pi$, 
as the probability that the
execution with policy $\pi$ is inside the target set $
\mathcal{T}$ at time $N$ and stays within the safe set $
\mathcal{S}$ for all time up to $N$.
From~\cite{SummersAutomatica2010},
\begin{align}
    \hat{r}_{\bar{x}_0}^\pi( \mathcal{S}, \mathcal{T})&=
    \Prob_{\bX}^{\bar{x}_0,\pi}\left\{\bx_N\in \mathcal{T} \wedge\bx_k\in
    \mathcal{S}\ \forall k\in \mathbb{N}_{[0,N-1]}\right\}. \nonumber %
\end{align}

From~\cite[Def. 10]{SummersAutomatica2010}, a Markov policy $\pi^\ast$ is
a \emph{maximal reach-avoid policy in the terminal sense} if and only if it is
the optimal solution of Problem A, defined as 
\begin{align}
  \begin{array}{rl}
      \mbox{A}: & \hat{r}_{\bar{x}_0}^{\pi^\ast}( \mathcal{S},
\mathcal{T})=\sup_{\pi\in \mathcal{M}}\hat{r}_{\bar{x}_0}^\pi( \mathcal{S},
\mathcal{T})
  \end{array}\label{prob:A}
\end{align} 
The solution of Problem A is characterized via dynamic programming~\cite[Thm.
11]{SummersAutomatica2010}.  Define $\hat{V}_k^\ast: \mathcal{X} \rightarrow
[0,1],\ k\in \mathbb{N}_{[0,N]}$, by the backward recursion for $\bar{x}\in
\mathcal{X}$,
\begin{align}
    \hat{V}_N^\ast(\bar{x})&=1_{ \mathcal{T}}(\bar{x})\label{eq:VT} \\
    \hat{V}_k^\ast(\bar{x})&=\sup_{\bar{u}\in \mathcal{U}}1_{
\mathcal{S}}(\bar{x})\int_{ \mathcal{X}}
\hat{V}_{k+1}^\ast(\bar{y})Q(d\bar{y}\vert\bar{x},\bar{u}).\label{eq:Vt_recursionQ} 
\end{align}
Then, the optimal value to Problem A is $\hat{r}_{\bar{x}_0}^{\pi^\ast}(
\mathcal{S}, \mathcal{T})=\hat{V}_0^\ast(\bar{x}_0)$ for every $\bar{x}_0\in
\mathcal{X}$. 

\begin{lem}{\cite[Thm. 11]{SummersAutomatica2010}}\label{lem:exist}
    A sufficient condition for existence of a maximal Markov policy for
    Problem A is \begin{align}
        \mathcal{U}_k(\bar{x},\lambda)&=\{\bar{u}\in \mathcal{U}: \int_{ \mathcal{X}}
        \hat{V}_{k+1}^\ast(\bar{y})Q(\bar{y}\vert \bar{x},\bar{u})d\bar{y} \geq
    \lambda\}\label{eq:Usuff}
\end{align} 
and $\mathcal U_k$ is compact for all $\lambda \in \mathbb{R},\bar{x}\in \mathcal{X}$
and $k\in \mathbb{N}_{[0,N-1]}$.
\end{lem}

Lemma~\ref{lem:exist} assures universal measurability of $\hat{V}_k^\ast(\cdot)$,
and that the Markov policy $\pi^\ast$ consists of universally measurable
maps $\mu^\ast_k$ ~\cite[Thm. 1 proof]{AbateAutomatica2008}. However, evaluating
(\ref{eq:Usuff}) is difficult.
We propose alternative sufficient conditions, which are %
easier to evaluate, and guarantee Borel-measurability (stronger than
universal measurability~\cite[Defn.  7.20]{BertsekasSOC1978}). 

\subsection{Open-loop stochastic reach-avoid analysis}

With $\bar{U} = [\bar{u}^\top_0\ \bar{u}^\top_1\ \ldots\ \bar{u}^\top_{N-1}]^\top \in
\mathcal{U}^N$ and $\bW = [\bw^\top_0\ \bw^\top_1\ \ldots\
\bw^\top_{N-1}]^\top\in \mathcal{W}^N$, we obtain
\begin{align}
    \bX&=\bar{A}\bar{x}_0+\bar{H}\bar{U}+\bar{G}\bW. \label{eq:sys_concat}
\end{align}
The matrices $\bar{A}, \bar{H}, \bar{G}$ are given by specific combinations
of the matrices $A$ and $B$ (see~\cite[Sec. 2]{SkafTAC2010}). 

Consider an open-loop policy $\rho: \mathcal{X} \rightarrow \mathcal{U}^N$ which
provides an open-loop sequence of inputs $\rho(\bar{x}_0)$ for every initial
condition $\bar{x}_0$. Then $\bX$, defined in \eqref{eq:sys_concat}
under the action of $\rho(\bar{x}_0)$, lies in the probability space $(
\mathcal{X}^N, \mathscr{B}( \mathcal{X}^N),
\Prob_{\bX}^{\bar{x}_0,\rho(\bar{x}_0)})$, with %
$\Prob_{\bX}^{\bar{x}_0,\rho(\bar{x}_0)}$ defined using $Q$~\cite[Prop.
7.45]{BertsekasSOC1978}. Note that $\rho(\bar{x}_0)\not\in \mathcal{M}$, since
universally measurable maps $\mu_k(\cdot)$ are functions of $\bx_k$, 
not $\bar{x}_0$. Consequently, a Markov policy with $\mu_k(\cdot)$ as constants
is a special case of $\rho(\cdot)$.

In~\cite{LesserCDC2013}, the authors approximate Problem A, without establishing the
direction of approximation, with Problem B,%
\begin{align}
  \begin{array}{rl}
      \mbox{B:} & \begin{array}{rcl}
      \mbox{maximize}& & \hat{r}_{\bar{x}_0}^{\rho(\bar{x}_0)}(\mathcal{S}, \mathcal{T}) \\
      \mbox{subject to}& & \left\{\begin{array}{rl}
        \bX&\sim \Prob_{\bX}^{\bar{x}_0,\rho(\bar{x}_0)}\\
        \rho(\bar{x}_0)&\in \mathcal{U}^N\\
    \end{array}\right.
  \end{array}\nonumber 
  \end{array}
\end{align}
with decision variable $\rho(\bar{x}_0)$, and
\begin{align}
 \hat{r}_{\bar{x}_0}^{\rho(\bar{x}_0)}(\mathcal{S},
\mathcal{T})&=
    \Prob_{\bX}^{\bar{x}_0,\rho(\bar{x}_0)}\left\{\bx_N\in \mathcal{T} \wedge\bx_k\in
    \mathcal{S}\ \forall k\in \mathbb{N}_{[0,N-1]}\right\}. \nonumber
\end{align}
The optimal solution to Problem B is $ \rho^\ast(\bar{x}_0)$. Since
$\rho(\bar{x}_0)\not\in \mathcal{M}$, the relation between Problems A and B,
apart from structural similarity, is not evident.  Problem B was solved
in~\cite{LesserCDC2013} approximately via particle filter and chance-constrained
optimization methods.

We first demonstrate that Problem B underapproximates Problem A, then use a
Fourier transform-based approach that enables an exact solution to Problem B.
\begin{prob} Characterize the sufficient conditions under which $
    \hat{V}_k^\ast(\cdot)$ and $\mu_k^\ast(\cdot)$ are
    Borel-measurable for the terminal time problem.\label{prob:Borel}
\end{prob}
\begin{prob} 
    Show that the terminal time 
    problem (Problem A) is
    underapproximated by the open-loop formulation (Problem
    B).\label{prob:underapprox}
\end{prob}
\begin{prob}
    a) Construct a scalable method for solving Problem B by characterizing the
    forward stochastic reach probability density for stochastic linear systems
    controlled by $\rho(\cdot)$ when $\bw$ has an arbitrary PDF. Additionally,
    b) formulate Problem B as a convex optimization problem when $\bw$ is
    Gaussian.\label{prob:FT}
\end{prob}

\section{Theoretical results}
\label{sec:theorRes}

\subsection{Sufficient conditions for Borel-measurability of $\hat{V}_k^\ast(\cdot)$}

\begin{defn}{~\cite[Defn. 7.12]{BertsekasSOC1978}}
A stochastic kernel $Q(\cdot\vert \bar{x},\bar{u})$ is continuous if for every
$(\bar{x},\bar{u}) \in \mathcal{X}\times \mathcal{U}$ and every sequence
$(\bar{x}_i,\bar{u}_i)\xrightarrow{ i \rightarrow \infty} (\bar{x},\bar{u})$, 
\begin{align}
    \lim_{i \rightarrow
\infty} Q(d\bar{y} \vert
(\bar{x}_i,\bar{u}_i))&=Q(d\bar{y} \vert (\bar{x},\bar{u})).
\end{align}
\end{defn}
\begin{lem}
   If the PDF of the disturbance $\bw$ $\psi_{\bw}$ is continuous, then
   $Q(\cdot\vert \bar{x},\bar{u})$ defined in \eqref{eq:stoch_kernel} is
   continuous.\label{lem:LTIstoch}
\end{lem}
Lemma~\ref{lem:LTIstoch} follows from the fact that continuity is preserved by
composition~\cite[Cor. 13.1.7]{TaoAnalysisII}. We have the following
theorem, similar to~\cite[Prop. 3]{DingAutomatica2013}.
\begin{thm}\label{thm:Borel}
    If $ \mathcal{U}$ is compact and $Q(\cdot\vert \bar{x}, \bar{u})$ is
    continuous, then $\hat{V}_k^\ast(\cdot)$ are Borel-measurable functions for
    $k\in \mathbb{N}_{[0,N]}$ and $\pi^\ast$, comprised of Borel-measurable maps
    $\mu_k^\ast(\cdot)$, exists.
\end{thm}
\begin{proof} (By induction)
    Since $\mathcal{S}, \mathcal{T}$ are Borel sets, $1_{\mathcal{S}}(\cdot)$
    and $1_{\mathcal{T}}(\cdot)$ are Borel-measurable functions, and the result for
    $k=N$ follows trivially. Consider the base case $k=N-1$. Since
    $\hat{V}_{N}^\ast(\cdot)$ is a bounded Borel-measurable function and
    $Q(\cdot\vert \bar{x},\bar{u})$ is a Borel-measurable function, continuous
    over $ \mathcal{X}\times \mathcal{U}$, $\int_{ \mathcal{X}}
    \hat{V}_{N}^\ast(\bar{y})Q(d\bar{y}\vert\bar{x},\bar{u})$ is continuous over
    $ \mathcal{X}\times \mathcal{U}$~\cite[Fact 3.9]{NowakAP1985}. Since
    continuity implies upper semi-continuity~\cite[Lem.
    7.13 (b)]{BertsekasSOC1978} and Borel-measurablity~\cite[Sec.
    1.4]{ChowProbability1997}, and $ \mathcal{U}$ is compact, an optimal
    Borel-measurable input map $\mu_{N-1}^\ast(\cdot)$ exists
    and $\int_{ \mathcal{X}}
    \hat{V}_{N}^\ast(\bar{y})Q(d\bar{y}\vert\bar{x},\mu_{N-1}^\ast(\bar{x}))$
    is Borel-measurable over $ \mathcal{X}$~\cite[Thm. 2]{HimmelbergMOR1976}.
    Finally, $\hat{V}^\ast_{N-1}(\cdot)$ is Borel-measurable since the product operator
    preserves Borel-measurability~\cite[Cor. 18.5.6]{TaoAnalysisII}.
    For the case $k=t$, assume for induction that $\hat{V}_{t+1}^\ast(\cdot)$
    is Borel-measurable.  By the same arguments as above, a
    Borel-measurable $\mu_t^\ast(\cdot)$ exists and $\hat{V}_{t}^\ast(\cdot)$ is
    Borel-measurable, completing the proof.
\end{proof}

Theorem~\ref{thm:Borel}
addresses Problem~\ref{prob:Borel}. Since Borel-measurability implies universal
measurability~\cite[Defn. 7.20]{BertsekasSOC1978}, 
the hypotheses of Theorem~\ref{thm:Borel} is stricter than
Lemma~\ref{lem:exist},
but can be easily checked, and  
implies that $\hat{V}_k^\ast(\bx_k)$ is a
$\mathscr{B}([0,1])$-measurable random
variable $\forall k\in \mathbb{N}_{[0,N]}$. 
The continuity requirements in Theorem~\ref{thm:Borel} and
Lemma~\ref{lem:LTIstoch} may be weakened to include exponential
densities~\cite[Sec.  8.3]{BertsekasSOC1978}.

\subsection{Problem B underapproximates Problem A}
\label{sub:consv}

Next, we address Problem~\ref{prob:underapprox}. For $\bx_k=\bar{x}$, denote the
expectation defined by $Q(\cdot\vert\bar{x},\bar{u})$ as $
\Exp_{\bx}^{\bar{u}}$. Under the conditions proposed by Theorem~\ref{thm:Borel},
we know that $\hat{V}_{k+1}^\ast(\bx_{k+1})$ is a Borel-measurable random
variable for all $k\in \mathbb{N}_{[0,N]}$. From \eqref{eq:Vt_recursionQ}, for
any $k\in \mathbb{N}_{[0,N-1]}$, we have \emph{almost surely} (a.s.)\footnote{
The a.s. equality arises because the conditional
expectation of $\hat{V}_{k+1}^\ast(\bx_{k+1})$ is defined only within an
equivalence (can differ in sets of zero probability measure)~\cite[Ch.
7]{ChowProbability1997}.}
\begin{align}
    \hat{V}_k^\ast(\bar{x})&=\sup_{\bar{u}\in \mathcal{U}}1_{
\mathcal{S}}(\bar{x}) \Exp_{\bx}^{\bar{u}}\left[\hat{V}_{k+1}^\ast(\bx_{k+1})
\middle\vert \bx_k=\bar{x} \right]. \label{eq:Vt_condexp}
\end{align}
Using Theorem~\ref{thm:Borel} and properties of conditional
expectations, we can show the following theorem.  See~\cite{VinodArxiv2017LCSS}
for the proof.

Note that the state $\bx_k$ is not an independent random vector, but part of a
Markov control process controlled by a sequence of actions.  Therefore, the
conditional expectation in \eqref{eq:Vt_condexp} is defined on the
$\sigma$-algebra $\sigma(\overline{\bx_k})=\sigma(\bx_0,\bx_1,\ldots,\bx_k)$.
For any $k\in \mathbb{N}_{[0,N-1]}$, we have from \eqref{eq:Vt_condexp}
\begin{align}
 \hat{V}_k^\ast(\bx_k)&=\sup_{\bar{u}\in \mathcal{U}}1_{
\mathcal{S}}(\bx_k) \Exp_{\bx}^{\bar{u}}\left[\hat{V}_{k+1}^\ast(\bx_{k+1})
\middle\vert \sigma(\overline{\bx_k}) \right]\mbox{a.s. }. \label{eq:Vt_rv}
\end{align}
Also, by the definition of a stochastic process~\cite[Sec. 5.3]{ChowProbability1997},
\begin{align}
    \sigma(\overline{\bx_k})\subset\sigma(\overline{\bx_{k+1}})\ \forall k\in
    \mathbb{N}_{[0,N-1]}.\label{eq:subalg}
\end{align}
\begin{thm}\label{thm:consv}
    If $ \mathcal{U}$ is compact and $Q(\cdot\vert \bar{x}, \bar{u})$ is
    continuous, 
    then $\hat{r}_{\bar{x}_0}^{\rho^\ast(\bar{x}_0)}(
    \mathcal{S}, \mathcal{T})\leq \hat{r}_{\bar{x}_0}^{\pi^\ast}( \mathcal{S},
    \mathcal{T})\mbox{ a.s. in }\bar{x}_0\in \mathcal{X}.$
\end{thm}
\begin{proof}
    For notational brevity, given $a\in \mathbb{N}_{[0,N]},\ b\in \mathbb{N}_{[0,N-1]}$, we
    define $U_a^{b} \deff %
    (\bar{u}_a,\bar{u}_{a+1},\ldots,\bar{u}_b)$ for some $U\in \mathcal{U}^N$
    with $U_0^{N-1}=U$ and $U_N^{N-1}$ as empty. We will later use a similar definition for $\bX$. 

    For $k\in \mathbb{N}_{[0,N-1]}$, define $\hat{W}_k: \mathcal{X}\times
    \mathcal{U}^{N-k} \rightarrow [0,1]$ based on \eqref{eq:Vt_condexp}, (a.s.)
    \begin{align}
        \hat{W}_{k}(\bar{x},U_k^{N-1})&\triangleq1_{ \mathcal{S}}(\bar{x})
        \Exp_{\bx}^{\bar{u}_{k}}\left[ \hat{W}_{k+1}(\bx_{k+1},U_{k+1}^{N-1}) \middle\vert
        \bx_{k}=\bar{x} \right] \nonumber\\
        \hat{W}_{N}(\bar{x},U_N^{N-1})&\triangleq1_{ \mathcal{T}}(\bar{x}). \nonumber
    \end{align}
    We see that $\hat{W}_k(\cdot)$ are Borel measurable by a straight-forward
    proof by induction.
    Next, we prove: 
    \begin{align}
        \begin{array}{ll}
        \mbox{S1:}& \hat{W}_k(\bar{x},U_k^{N-1})\leq \hat{V}_k^\ast(\bar{x})\mbox{ a.s. in
    }\bar{x}\in \mathcal{X}, \forall k\in \mathbb{N}_{[0,N-1]} \nonumber \\
        \mbox{S2:}& \sup_{\rho(\bar{x}_0)\in \mathcal{U}^N} \hat{W}_0(\bar{x}_0,\rho(\bar{x}_0))=\hat{r}_{\bar{x}_0}^{\rho^\ast(\bar{x}_0)}(
    \mathcal{S}, \mathcal{T})\mbox{ a.s. in }\bar{x}_0. \nonumber
        \end{array}
    \end{align}
    S1 implies $\sup_{\rho(\bar{x}_0)\in \mathcal{U}^N} \hat{W}_0(\bar{x}_0,\rho(\bar{x}_0))\leq
    \hat{V}_0^\ast(\bar{x}_0)$ a.s. in
    $\bar{x}_0\in \mathcal{X}$. Thus, the proof is complete via S2 and the fact
    that $\hat{V}_0^\ast(\bar{x}_0)=\hat{r}_{\bar{x}_0}^{\pi^\ast}( \mathcal{S},
    \mathcal{T})$ for every $\bar{x}_0\in \mathcal{X}$.

    \emph{Proof of S1}: (By induction) The base case is $k=N-1$.
    Since $U_{N-1}^{N-1}=\bar{u}_{N-1}$,
    \begin{align}
        \hat{W}_{N-1}(\bar{x},\bar{u}_{N-1})&\triangleq1_{ \mathcal{S}}(\bar{x})
        \Exp_{\bx}^{\bar{u}_{N-1}}\left[ 1_{ \mathcal{T}} (\bx_N) \middle\vert
        \bx_{N-1}=\bar{x} \right]\mbox{a.s. }. \nonumber
    \end{align}
    From \eqref{eq:Vt_condexp} and \eqref{eq:VT}, we have (a.s.)
    \begin{align}
        \hat{W}_{N-1}(\bar{x},\bar{u}_{N-1})&\leq \sup_{\bar{u}_{N-1}\in
        \mathcal{U}}\hat{W}_{N-1}(\bar{x},\bar{u}_{N-1})=
        \hat{V}_{N-1}^\ast(\bar{x}). \nonumber
    \end{align}
    For the case $k=t$, assume for induction that
    $\hat{W}_{t+1}(\bar{x},U_{t+1}^{N-1})\leq \hat{V}_{t+1}^\ast(\bar{x})$
    a.s.  in $\bar{x}\in \mathcal{X}$. We have to show
    $\hat{W}_{t}(\bar{y},U_{t}^{N-1})\leq \hat{V}_{t}^\ast(\bar{y})$ a.s.\ in
    $\bar{y}\in \mathcal{X}$. By Property P1 and \eqref{eq:Vt_rv}, we have
    (a.s.)
    \begin{align}
        \hat{W}_{t}(\bar{y},U_{t}^{N-1}) 
                          &\leq \sup_{\bar{u}_{t}\in
        \mathcal{U}}\hat{W}_{t}(\bar{y},U_{t}^{N-1}) \nonumber \\
        &\leq \sup_{\bar{u}_{t}\in
        \mathcal{U}}1_{ \mathcal{S}}(\bar{y})
        \Exp_{\bx}^{\bar{u}_{t}}\left[ \hat{W}_{t+1}(\bx_{t+1},U_{t+1}^{N-1})\middle\vert
            \bx_{t}=\bar{y} \right] \nonumber \\
        &\leq \sup_{\bar{u}_{t}\in
        \mathcal{U}}1_{ \mathcal{S}}(\bar{y})
        \Exp_{\bx}^{\bar{u}_{t}}\left[ \hat{V}^\ast_{t+1}(\bx_{t+1})\middle\vert
            \bx_{t}=\bar{y} \right] \nonumber \\
        &= \hat{V}_{t}^\ast(\bar{y}). \nonumber
    \end{align}
    This completes the proof of S1.

    \emph{Proof of S2}: 
    We use the definition of $\hat{W}_k(\cdot)$ based on $\sigma$-algebra
    (similar to \eqref{eq:Vt_rv}).  By~\eqref{eq:subalg}, $1_\mathcal{S}(\bx_k)$
    is a $\sigma(\overline{\bx_t})$-measurable random variable for $t\in
    \mathbb{N}_{[k,N]}$. Expanding $\hat{W}_0(\cdot)$ and $\hat{W}_1(\cdot)$ and
    using Property P2, we have (a.s.)
    \begin{align}
        \hat{W}_0(\bx_{0},U)&=1_\mathcal{S}(\bx_0)\Exp_{\bx}^{\bar{u}_{0}}\left[ \hat{W}_1
        (\bx_1,U_{1}^{N-1}) \middle\vert
        \sigma(\bx_{0}) \right] \nonumber \\
            &=\Exp_{\bx}^{\bar{u}_{0}}\left[1_\mathcal{S}(\bx_0) \hat{W}_1
        (\bx_1,U_{1}^{N-1}) \middle\vert
        \sigma(\bx_{0}) \right] \nonumber \\ %
            &=\Exp_{\bx}^{\bar{u}_{0}}\Big[1_\mathcal{S}(\bx_0)
        1_\mathcal{S}(\bx_1) \nonumber \\
    &\quad\quad\left.\Exp_{\bx}^{\bar{u}_{1}}\left[ \hat{W}_2 (\bx_2,U_{2}^{N-1})
\middle\vert \sigma(\overline{\bx_{1}}) \right] \middle\vert \sigma(\bx_{0})
\right] \nonumber \\% \label{eq:rec2} \\
&=\Exp_{\bx}^{\bar{u}_{0}}\Big[\Exp_{\bx}^{\bar{u}_{1}}\Big[ 1_\mathcal{S}(\bx_0)
    1_\mathcal{S}(\bx_1) \nonumber \\
    &\quad\quad\left.\left.\hat{W}_2 (\bx_2,U_{2}^{N-1})
\middle\vert \sigma(\overline{\bx_{1}}) \right] \middle\vert \sigma(\bx_{0})
\right] \nonumber %
    \end{align}
    From Property P3 and~\eqref{eq:subalg}, we have (a.s.)
    \begin{align}
        \hat{W}_0(\bx_{0},U)=\Exp_{\bX_1^2}^{U_0^1}\left[\left(\prod_{k=0}^{1} 1_{ \mathcal{S}}(\bx_k)\right)\hat{W}_2 (\bx_2,U_{2}^{N-1})
        \middle\vert \sigma(\bx_{0}) \right] \nonumber
    \end{align}
    We repeatedly expand $\hat{W}_k(\cdot)$ for $k\in \mathbb{N}_{[2,N]}$ and
    apply the arguments presented above to obtain
    \begin{align}
        \hat{W}_0(\bx_{0},U)= \Exp_{\bX_1^N}^{U_0^{N-1}}&\left[\left(\prod_{k=0}^{N-1} 1_{ \mathcal{S}}(\bx_k)\right) 1_{
\mathcal{T}}(\bx_N) \middle\vert \sigma(\bx_0)\right]\mbox{ a.s.} \nonumber
    \end{align}
    By definition of $\Prob_{\bX}^{\bar{x}_0,\rho(\bar{x}_0)}$,
    ($\bX_1^N=\bX,U_0^{N-1}=U\triangleq\rho(\bar{x}_0)$)
    \begin{align}
        \hat{W}_0(\bar{x}_{0},\rho(\bar{x}_0))=\hat{r}_{\bar{x}_0}^{\rho(\bar{x}_0)}(
\mathcal{S}, \mathcal{T})\mbox{ a.s.}, \nonumber
    \end{align}
    and the definition of $\rho^\ast(\bar{x}_0)$ completes the proof of S2.
\end{proof}

We denote the optimal value of Problem B as $\hat{W}_0^\ast(\bar{x}_0)$. 

\section{Under-approximation via Fourier transforms}
\label{sec:solveProbB}

\subsection{FTBU using an analytical expression for $\hat{r}_{\bar{x}_0}^{\rho(\bar{x}_0)}(\mathcal{S}, \mathcal{T})$}
\label{sub:an_analytical_expression_for_bx_cdot_}

Let the PDF of the random vector $\bX$ parameterized by the initial condition
$\bar{x}_0$ and the input vector $\bar{U}$ be $\psi_{\bX}(\bar{X};\bar{x}_0,\bar{U})$.
The objective of Problem B
$\hat{r}_{\bar{x}_0}^{\rho(\bar{x}_0)}(\mathcal{S}, \mathcal{T})$ is
\begin{align}
\hat{r}_{\bar{x}_0}^{\rho(\bar{x}_0)}(\mathcal{S}, \mathcal{T})
    &=\int_{
\mathcal{T}}\underbrace{\int_{ \mathcal{S}}\ldots\int_{
\mathcal{S}}}_{\text{$N-1$ times}}\psi_{\bX}(\bar{X};\bar{x}_0,\rho(\bar{x}_0))d\bar{X}\label{eq:obj_probB}
\end{align}
where $\bar{X}={[\bar{x}_1^\top\ \bar{x}_2^\top\ \ldots\
\bar{x}_{N}^\top]}^\top\in \mathcal{X}^N$, $\bar{x}_k\in \mathcal{X}\ \forall
k\in \mathbb{N}_{[1,N]}$, and $d\bar{X}$ is short for
$d\bar{x}_1d\bar{x}_2\ldots d\bar{x}_N$. Therefore, if
$\psi_{\bX}$ is known, then  $\hat{r}_{\bar{x}_0}^{\rho(\bar{x}_0)}(\mathcal{S},
\mathcal{T})$ is a $nN$-dimensional integral of a PDF $\psi_{\bX}$ over
$ \mathcal{S}\times \mathcal{S} \times \ldots \times \mathcal{T}$.
Determining $\psi_{\bX}$ for a known $\bar{U}$ can be posed as a forward
stochastic reachability problem %
using the CF of $\bW$~\cite[Prop. P3]{VinodHSCC2017} defined as
\begin{align}
    \Psi_{\bW}(\bar{\alpha})&=\prod_{k=0}^{N-1}\Psi_{\bw}(\bar{\alpha}_k)\label{eq:cfun_W}
\end{align}
where $\bar{\alpha}={[\bar{\alpha}_0^\top\ \bar{\alpha}_1^\top\ \ldots\
\bar{\alpha}_{N-1}^\top]}^\top\in \mathbb{R}^{(nN)}$ and $\bar{\alpha}_k\in
\mathbb{R}^n$ for all $k\in \mathbb{N}_{[0,N-1]}$.  
We compute $\psi_{\bX}$ via Proposition~\ref{prop:FSRPD_U}.

\begin{prop}
    For initial state $\bar{x}_0\in \mathcal{X}$, dynamics as in
    \eqref{eq:sys_concat}, and open-loop control vector $\bar{U}$, the PDF and
    CF of $\bX$ are
\begin{align}
    \Psi_{\bX}(\bar{\beta};\bar{x}_0,\bar{U})&=\exp(j\bar{\beta}^\top
    (\bar{A}\bar{x}_0+\bar{H}\bar{U}))\Psi_{\bW}(\bar{G}^\top\bar{\beta})\label{eq:cfun_X}\\
    \psi_{\bX}(\bar{X};\bar{x}_0,\bar{U})&=\mathscr{F}^{-1}\left\{\Psi_{\bX}(\bar{\beta};\bar{x}_0,\bar{U})\right\}(-\bar{X})\label{eq:pdf_X}
\end{align}
where $\bar{\beta}={[\bar{\beta}_1^\top\ \bar{\beta}_2^\top\ \ldots\
\bar{\beta}_{N}^\top]}^\top\in \mathbb{R}^{(nN)}$ and $\bar{\beta}_k\in
\mathbb{R}^n$ for all $k\in \mathbb{N}_{[1,N]}$. \label{prop:FSRPD_U}
\end{prop}
\begin{proof}
From \eqref{eq:sys_concat}, \eqref{eq:cfun_W}, and~\cite[Property P2]{VinodHSCC2017}.
\end{proof}

In general, \eqref{eq:pdf_X} is a $nN$-dimensional integration 
\eqref{eq:cfun_ift}. However,
when the CF of $\bX$ is in a standard form, a closed-form
expression for $\psi_{\bX}(\cdot)$ can be obtained, and we compute
$\hat{r}_{\bar{x}_0}^{\rho(\bar{x}_0)}(\mathcal{S}, \mathcal{T})$ via
\eqref{eq:obj_probB}. Else, we can compute
$\hat{r}_{\bar{x}_0}^{\rho(\bar{x}_0)}(\mathcal{S}, \mathcal{T})$ using
$\Psi_{\bX}$ if the Fourier transform of $1_{\mathcal{S}}(\cdot)$ and
$1_{\mathcal{T}}(\cdot)$ is known and $ \Exp\left[  \bX^\top \bX \right]<\infty$
\cite[Sec. 4.2]{VinodHSCC2017}. Thus, we can solve Problem B, and thereby
Problem~\ref{prob:FT}a, using Proposition~\ref{prop:FSRPD_U} and \eqref{eq:obj_probB}
for arbitrary $\psi_{\bw}$.

Note that while scalability of this approach is contingent on high-dimensional
quadrature, 
this challenge is far more tractable than the computational and memory costs
associated with DPBDA. In general, we can compute \eqref{eq:obj_probB} for
arbitrary disturbance densities through Monte-Carlo simulations~\cite[Sec.
4.8]{PressNumerical2007},~\cite[Ch. 4.2.1]{genz2009computation} and quasi-Monte Carlo
simulations~\cite[Ch. 4.2.2]{genz2009computation}. 

\subsection{Gaussian disturbance}
\label{sub:GaussGenz}

When $\bw$ is a Gaussian random vector, the CF of $\bw\sim \mathcal{N}(\bar{m},\Sigma)$  ~\cite[Sec. 9.3]{GubnerProbability2006} is
\begin{align}
    \Psi_{\bw}(\bar{\alpha})&=\exp\left(j\bar{\alpha}^\top
\bar{m}-\frac{\bar{\alpha}^\top\Sigma\bar{\alpha}}{2}\right).\label{eq:cfun_gauss}
\end{align}
Using \eqref{eq:cfun_W}, \eqref{eq:cfun_gauss}, and
Proposition~\ref{prop:FSRPD_U}, $\psi_{\bX}(\cdot)$ is described by
\begin{subequations}
\begin{align}
    \bX&\sim \mathcal{N}\big(\bar{m}_{\bX},\Sigma_{\bX}\big)\label{eq:gaussX_pdf}\\
    \bar{m}_{\bX}&=\bar{G}(\bar{1}_{N\times 1}\otimes
    \bar{m})+\bar{A}\bar{x}_0+\bar{H}\bar{U} \label{eq:gaussX_mean} \\
    \Sigma_{\bX}&=\bar{G}(I_{N}\otimes \Sigma)\bar{G}^\top. \label{eq:gaussX_sigma}
\end{align}\label{eq:gaussX}%
\end{subequations}%

\begin{prop}\label{prop:logconcaveProbB}
    For convex $ \mathcal{U}$, $ \mathcal{S}$, and $ \mathcal{T}$, dynamics as in
    \eqref{eq:sys} and a Gaussian disturbance $\bw$, Problem B is log-concave.
\end{prop}
\begin{proof}
    From~\cite[Sec. 2.3]{DharmadhikariUnimodality1988}, $\by\sim
    \mathcal{N}(0,\Sigma_{\bX})$ is log-concave with respect to $\by$.
    By~\cite[Sec.  3.2.2]{BoydConvex2004}, $\psi_{\bX}$ \eqref{eq:gaussX_pdf}
    is log-concave in $\bar{U}$ since it is an affine transformation
    of $\psi_{\by}$ by $\by-\bar{m}_X$.  
    From~\cite[Sec. 2.3.2, Sec. 3.5.2]{BoydConvex2004}, 
    sets $ \mathcal{T}\times \mathcal{S}\times\ldots\times \mathcal{S}$ and $
    \mathcal{U}^N$ are convex, and
    $\hat{r}_{\bar{x}_0}^{\rho(\bar{x}_0)}(\mathcal{S}, \mathcal{T})$ is
    log-concave over $ \mathcal{U}^N$. Thus, Problem B is log-concave.
\end{proof}

Proposition~\ref{prop:logconcaveProbB} addresses Problem~\ref{prob:FT}b.
For stochastic linear systems with a Gaussian disturbance and polytopic $
\mathcal{S}$ and $\mathcal{T}$, \eqref{eq:obj_probB} is the integration of a
multivariate Gaussian random variable over a polytope. 
Efficient computation of \eqref{eq:obj_probB} and
log-concavity (Proposition~\ref{prop:logconcaveProbB}) enables
a scalable solution to Problem B when $\bw$ is Gaussian. 

\subsection{FTBU implementation for the Gaussian disturbance case}
\label{sub:FTBUimplement}

To solve \eqref{eq:obj_probB} when $\bw$ is Gaussian, we use Genz's
algorithm~\cite{GenzCode}, which is based on quasi-Monte-Carlo simulations and
Cholesky decomposition~\cite{GenzJCGS1992}.  Genz's algorithm provides an error
estimate that is the result of a trade-off between accuracy and computation
time.  We set the number of particles for the Monte-Carlo simulation so that the
error estimate is less than some $\epsilon>0$. 
This results in a runtime evaluation of
$\hat{r}_{\bar{x}_0}^{\rho(\bar{x}_0)}(\mathcal{S}, \mathcal{T})$ that is
dependent on $\bar{x}_0$, unlike typical Monte-Carlo simulations. %
To take the logarithm of
$\hat{r}_{\bar{x}_0}^{\rho(\bar{x}_0)}(\mathcal{S}, \mathcal{T})$ in
Proposition~\ref{prop:logconcaveProbB}, we set
$\hat{r}_{\bar{x}_0}^{\rho(\bar{x}_0)}(\mathcal{S}, \mathcal{T})=\epsilon$
if $\hat{r}_{\bar{x}_0}^{\rho(\bar{x}_0)}(\mathcal{S},
\mathcal{T})<\epsilon$.

While the convexity result in Proposition~\ref{prop:logconcaveProbB} ensures a
tractable, globally optimal solution to Problem B, the lack of a closed-form
expression for the objective \eqref{eq:obj_probB} requires black-box optimization
techniques. Further, since Genz's algorithm enforces an accuracy of only
$\epsilon$, the log-concavity of $\hat{r}_{\bar{x}_0}^{\rho(\bar{x}_0)}(\mathcal{S},
\mathcal{T})$  
may not be preserved. %
Hence the ideal solver for Problem B should handle the ``noisy''
evaluation of \eqref{eq:obj_probB} as an oracle, and solve a constrained
optimization problem. 

We use MATLAB's
\emph{patternsearch} to solve
Problem B, because it is based on direct search
optimization~\cite{kolda2003optimization} and can handle estimation errors in
\eqref{eq:obj_probB} efficiently. The solver is a derivative-free optimizer and
uses evaluations over an adaptive mesh to obtain feasible descents towards the
globally optimal solution. However, it requires a larger number of function
evaluations as compared to \emph{fmincon}. For linearly-constrained and
bound-constrained optimization problems (such as Problem B, which 
is linearly constrained when $ \mathcal{U}$ is a polytope), 
creating the mesh using \emph{generating set search} reduces the number
of function evaluations~\cite[Sec.  8]{kolda2003optimization}.

\subsection{Advantages and limitations of FTBU}

The main advantage of FTBU is that it does not require gridding of the 
state, input, or disturbance spaces.
Unlike the DPBDA~\cite{AbateHSCC2007}, which solves Problem A on a grid over $
\mathcal{S}$, irrespective of the size of the initial set of interest, the FTBU
solves Problem B at a desired $\bar{x}_0$. By converting the terminal time problem into
an optimization problem involving a multi-dimensional integral, FTBU
achieves higher computational speed at lower memory cost 
(Figure~\ref{fig:curse_of_dim}) for a given initial condition. 
Probabilistically verifying a set of initial conditions
would require performing FTBU over a grid on the state space 
(Figure~\ref{fig:consv}), thereby losing any computational advantage over DPBDA. 
An alternative approach for the verification problem relies on Lagrangian methods~\cite{GleasonCDC2017}.
While evaluating \eqref{eq:obj_probB} can be
computationally expensive for arbitrary disturbances,
for Gaussian disturbances we can compute
\eqref{eq:obj_probB} efficiently (see Section~\ref{sub:FTBUimplement}). Further,
since the dimension of the integral in \eqref{eq:obj_probB} is $nN$, 
large $n$ effectively limits the time horizon $N$.  Additionally, the lack of
feedback in $\rho(\cdot)$ implies $N$ cannot be large~\cite{LesserCDC2013}, as it may 
induce excessive conservatism in the underapproximation.

\section{Numerical example}
\label{sec:numex}

Consider the discrete-time chain of integrators, with state $\bx_k\in \mathbb{R}^n$, input $u_k\in [-1,1]$, a Gaussian
disturbance $\bw_k \sim \mathcal{N}(0,0.01I_n)$, sampling time $N_s=10$, and
time horizon $N=10$. 
\begin{align}
    \bx_{k+1}&= \left[ {\begin{array}{ccccc} 
    1 & N_s & \frac{1}{2}N_s^2 & \ldots  & \frac{1}{(n-1)!} N_s^{n-1}  \\   
    0 & 1   & N_s              &         &                             \\  
    \vdots & &                 & \ddots  & \vdots                      \\
    0 & 0 & 0 & \ldots & N_s                      \\
    0 & 0 & 0 & \ldots & 1                      \\
    \end{array} } \right]\bx_k \nonumber \\
    &\quad+ {\left[ {\begin{array}{cccc} 
    \frac{1}{n!} N_s^n & \ldots & \frac{1}{2} N_s & N_s  \end{array} }
    \right]}^\top u_k +\bw_k\label{eq:sys_I_chain}
\end{align}

All computations were performed using MATLAB on an Intel Core i7
CPU with 3.4GHz clock rate and 16 GB RAM.  MATLAB code for this work is
available at \url{http://hscl.unm.edu/files/code/LCSS17.zip}.

\subsection{Comparison of FTBU and DPBDA runtimes and bounds}

We first demonstrate 1) scalability of the underapproximation as compared to the
DPBDA, and 2) non-trivial lower bounds obtained using FTBU.
Figure~\ref{fig:curse_of_dim} shows how FTBU and DPBDA scale with state
dimension $n$, for $n \leq 40$. We solve Problems A and B with 
$ \mathcal{S}={[-10,10]}^n,\ \mathcal{T}={[-5,5]}^n$, and $\epsilon=0.01$.
For DPBDA,
we restrict the grid over $\mathcal{X}$ to $\mathcal{S}$ for $n\leq 3$.  
We approximate the disturbance space as $[-0.5,0.5]^n$,
based on the covariance matrix of $\bw_k$.  We discretize
$ \mathcal{X}$, $\mathcal{W}$, and $\mathcal{U}$ with grid spacings of 
$0.05$, $0.05$, and $0.1$,
respectively.  
As expected, FTBU implemented using \emph{patternsearch}
is slower than \emph{fmincon}, but both implementations
scale with dimension $n$ much better than DPBDA.

\begin{figure}
    \centering 
\includegraphics[width=0.75\linewidth]{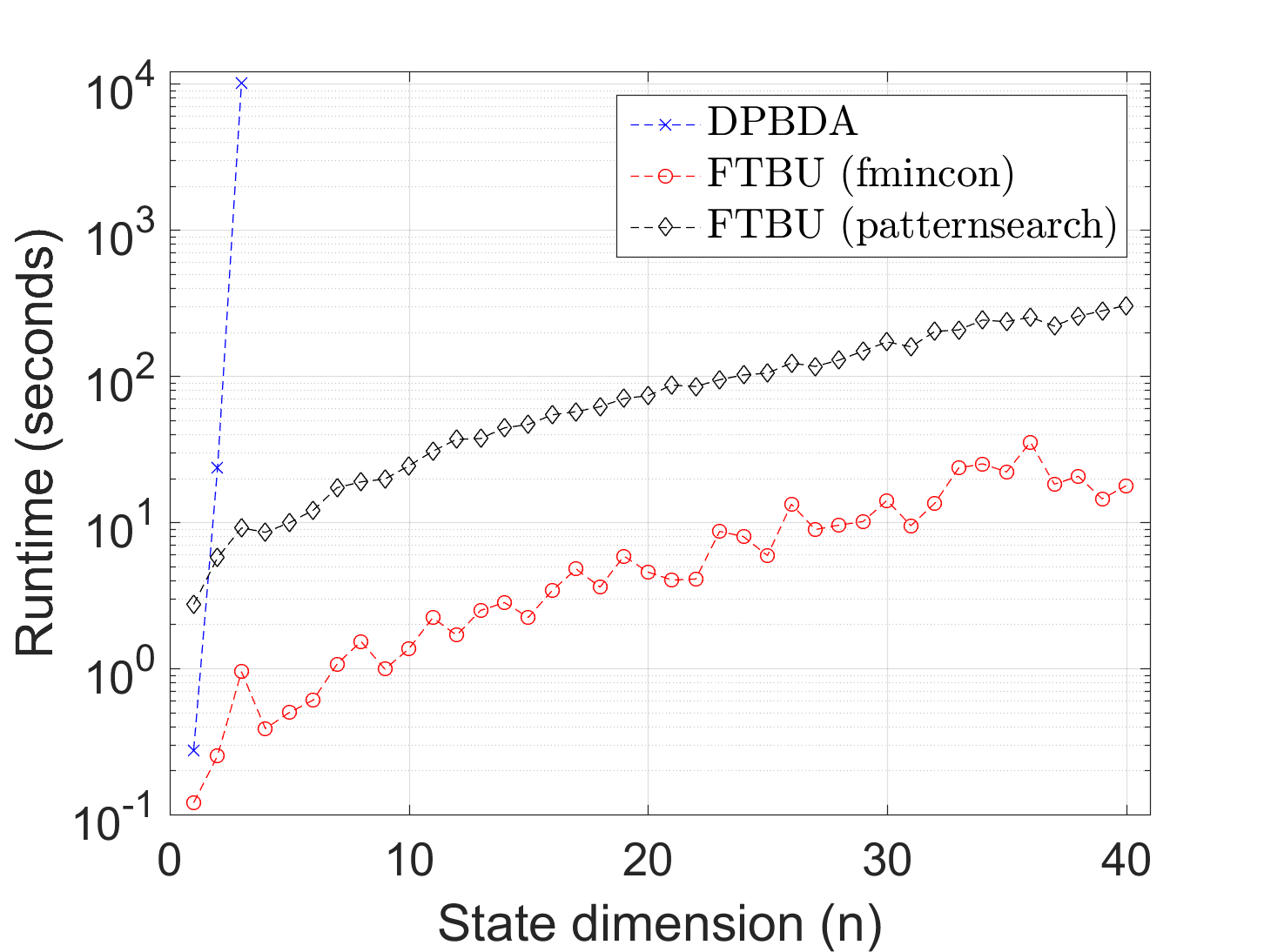} 
\caption{Scalability of DPBDA and FTBU with state dimension $n$ to compute
$\hat{V}_0^\ast(\bar{x}_0)$ and $\hat{W}_0^\ast(\bar{x}_0)$ for some
$\bar{x}_0\in \mathcal{X}$. Average computation time for the FTBU based on $20$
randomly chosen points in $ \mathcal{T}$ at each $n$.}
\label{fig:curse_of_dim}
\end{figure}

\begin{figure*}
\centering
\newcommand{\figw}{0.25}
\subfloat[]{\includegraphics[width=\figw\linewidth]{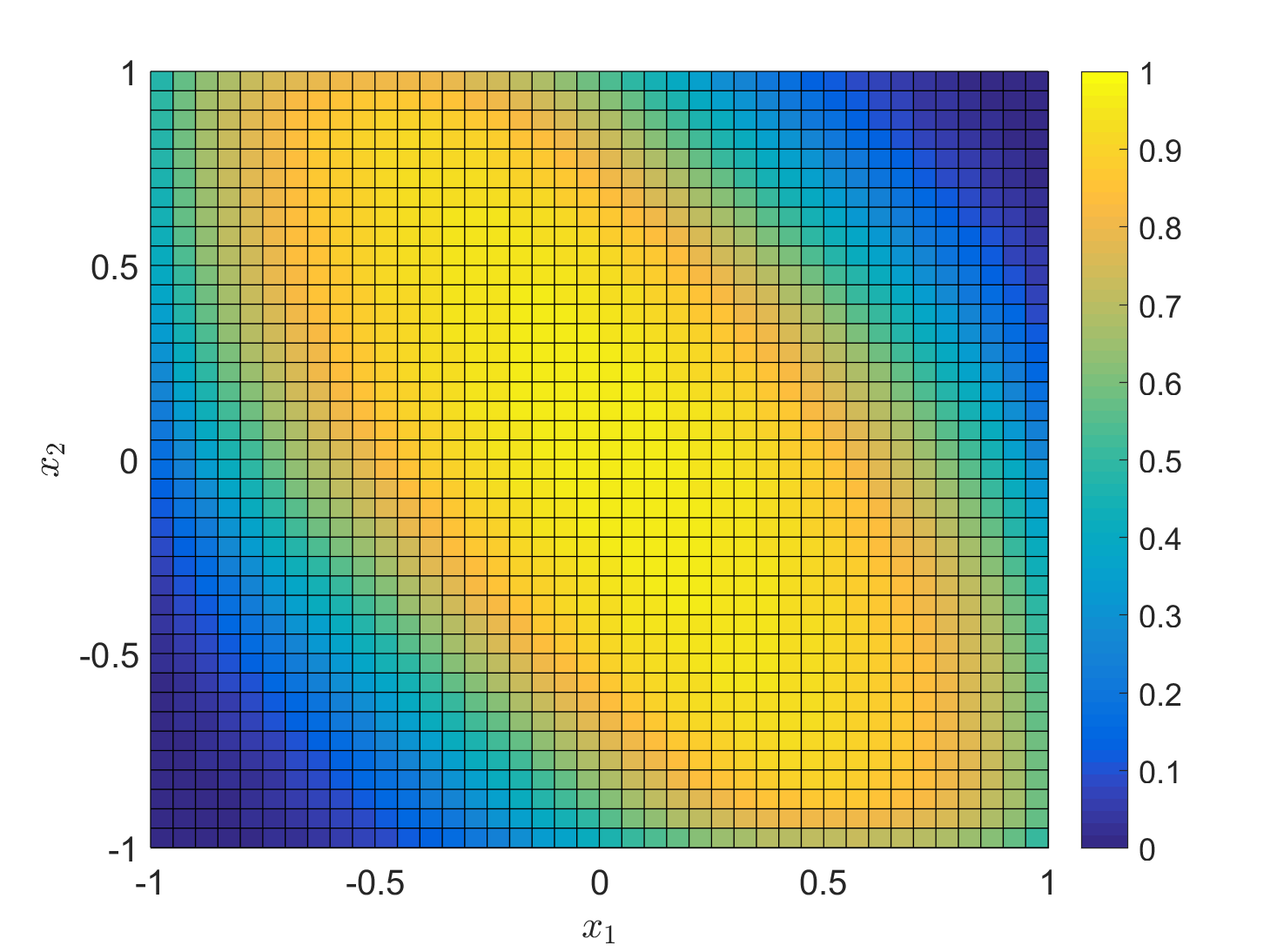}}
\qquad
\subfloat[]{\includegraphics[width=\figw\linewidth]{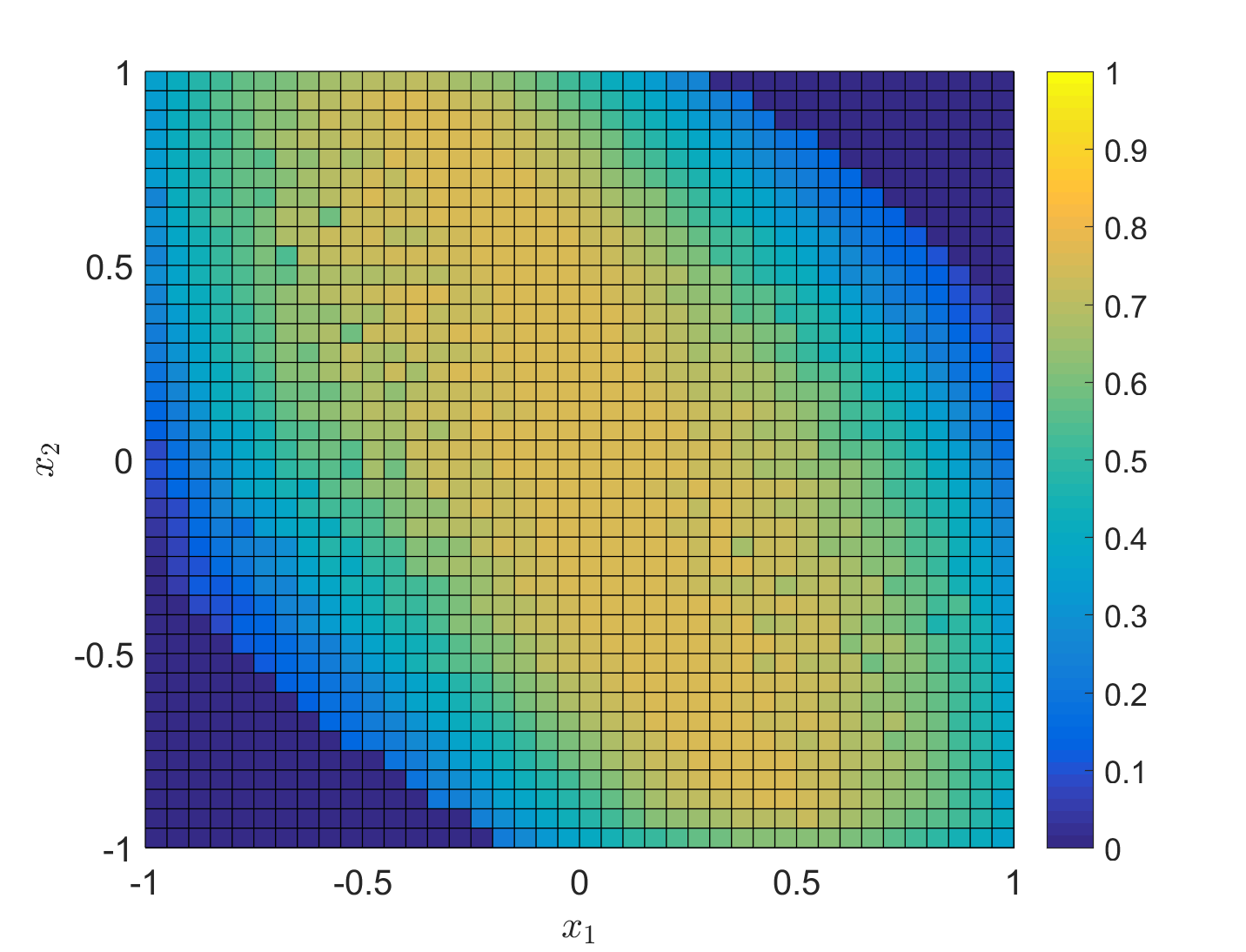}}
\qquad
\subfloat[]{\includegraphics[width=\figw\linewidth]{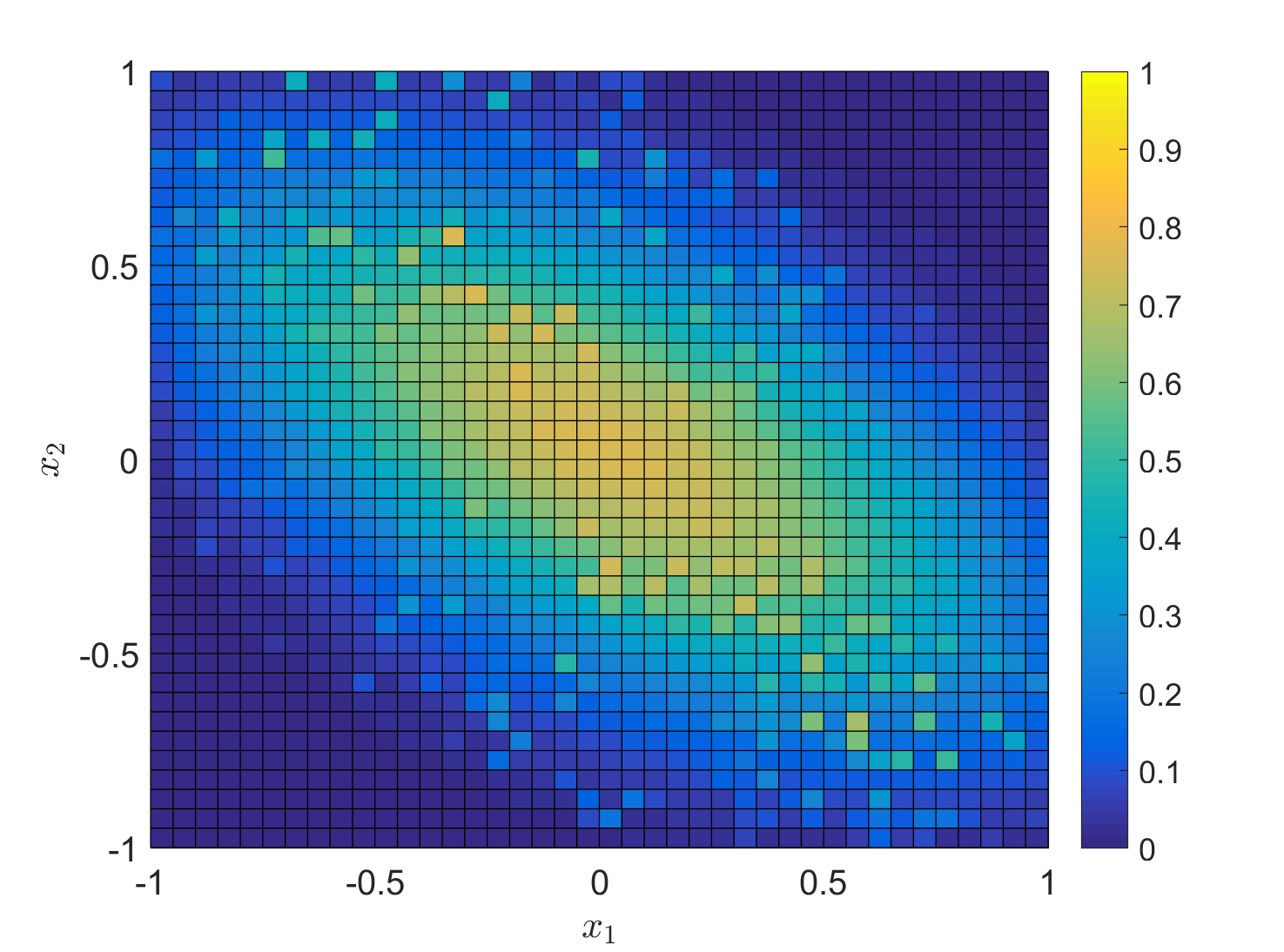}}\\
\subfloat[]{\includegraphics[width=\figw\linewidth]{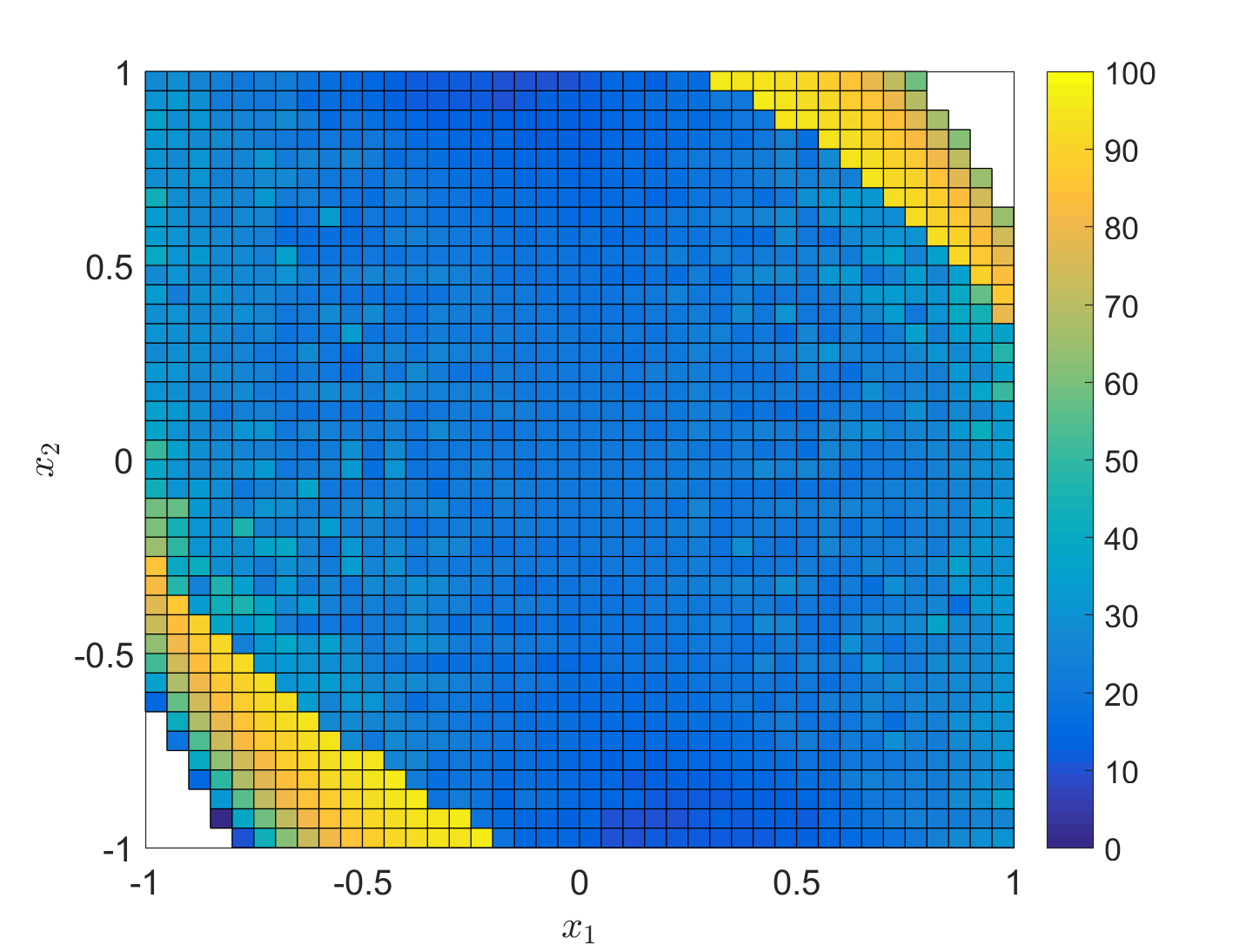}}
\qquad
\subfloat[]{\includegraphics[width=\figw\linewidth]{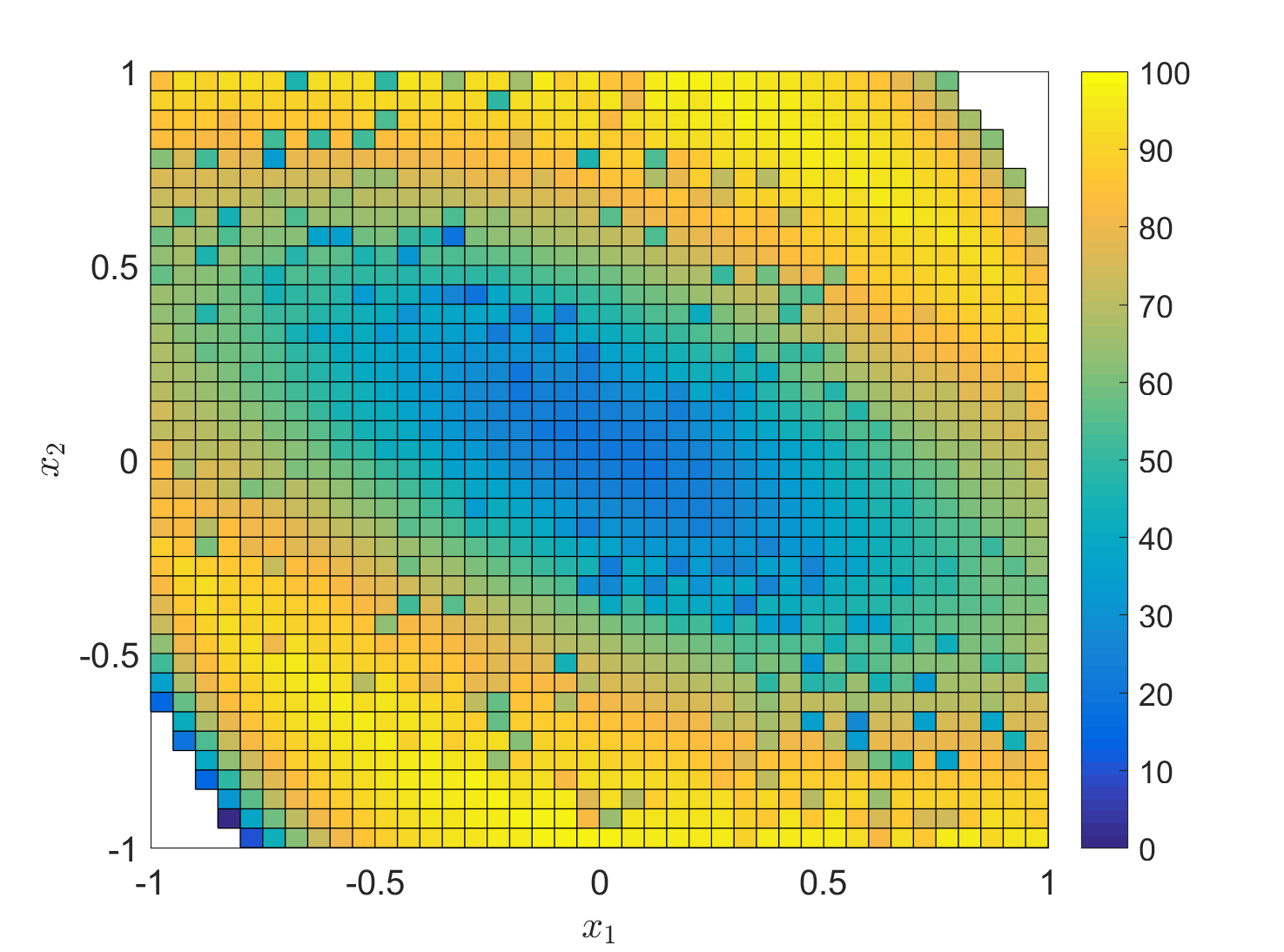}}
\qquad
\subfloat[]{\includegraphics[width=\figw\linewidth]{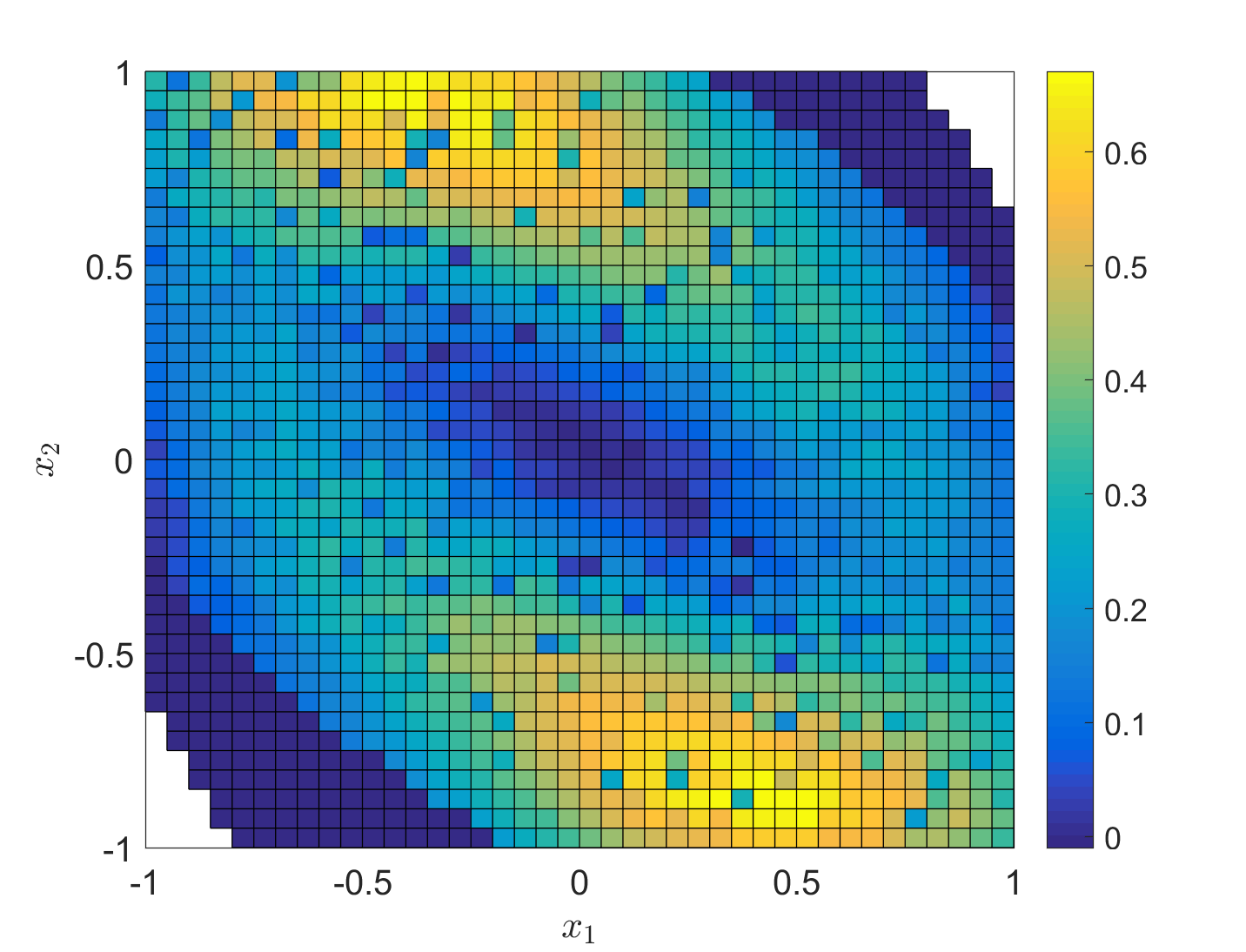}}\\
\caption{Conservativeness of the underapproximation for a double integrator
    ($n=2$). (a) Dynamic programming (DPBDA) approximation of $
    \hat{V}_0^\ast(\cdot)$;  Fourier transform-based underapproximation (FTBU) $
    \hat{W}_0^\ast(\cdot)$ computed (b) using \emph{patternsearch} and (c) using
    \emph{fmincon} for each grid point; Relative error defined as
    $\frac{\hat{V}_0^\ast(\cdot)-\hat{W}_0^\ast(\cdot)}{\hat{V}_0^\ast(\cdot)}\times100$
    for $\hat{V}_0^\ast(\cdot)>\epsilon$ and $ \hat{W}_0^\ast(\cdot)$ computed
    (d) using \emph{patternsearch} and (e) using  \emph{fmincon} for each grid
    point; Note that higher errors occur closer to the boundary, as expected,
    due to the lack of feedback; (f) Improvement in $\hat{W}_0^\ast(\cdot)$
    using \emph{patternsearch} instead of \emph{fmincon}; $\epsilon=0.01$
\label{fig:consv}}
\end{figure*}

Table~\ref{tab:test_points} summarizes the bounds on $\hat{V}_0^\ast(\cdot)$ for
various $\bar{x}_0$ at $n=40$.  For high $n$, \eqref{eq:sys_I_chain} becomes severely
under-actuated and the influence of the disturbance becomes very strong. This leads to the open-loop
formulation yielding trivial lower bounds,
$\hat{r}_{\bar{x}_0}^{\rho^\ast(\bar{x}_0)}(\mathcal{S}, \mathcal{T})=\epsilon$,
for many $\bar{x}_0$ in the original $\mathcal{S},\mathcal{T}$. We therefore set
$ \mathcal{T}={[-8,8]}^{40}$ and $\mathcal{S}={[-10,10]}^{40}$. 
While
Theorem 2 assures that $\hat{W}_0^\ast(\cdot)$ is a lower bound on
$\hat{V}_0^\ast(\cdot)$, this bound is subject to $\epsilon$, hence the
discrepancies between the numerical values for $\hat{W}_0^\ast(\cdot)$ and the
lower bounds on $\hat{V}_0^\ast(\cdot)$. We use $\epsilon=0.001$.

\begin{table}
    \centering
    \tabcolsep=0.11cm
    \begin{tabular}{|c|c|c|c|c|c|c|}
    \hline									%
    \multirow{2}{*}{\begin{minipage}{2cm}\centering
    Initial state of\\interest $\bar{x}_0^\top$ \end{minipage}} & \multirow{2}{*}{\begin{minipage}{1cm}\centering
    Member
            \\of\end{minipage}} & \multicolumn{2}{c|}{$\hat{W}_0^\ast(\bar{x}_0)$} &
    \multirow{2}{*}{\begin{minipage}{1cm}\centering
    $\hat{V}_0^\ast(\bar{x}_0)$\end{minipage}} & \multicolumn{2}{c|}{Runtime (s)}
            \\\cline{3-4}\cline{6-7}
            & & \emph{fm} & \emph{ps} & & \emph{fm} & \emph{ps}\\ \hline
            $[0\ 0\ 0\ \ldots\ 0]$ & $\mathcal{T}$ & $1$ & $1$ & $[0.999,1]$ &
            $12$ & $302$ \\ \hline
            $[2.5\ 2.5\ 2.5\ \ldots]$ & $\mathcal{T}$ & $0.984$ & $0.986$ &
            $[0.985,1]$ & $798$  & $1196$\\ \hline
    $[$-$8.5\ 8\ $-$8.5\ 8\ \ldots]$ & $ \mathcal{S}\setminus\mathcal{T}$
            & $0.500$ & $0.999$ & $[0.998,1]$ & $12$ & $441$ \\ \hline
   \end{tabular}
   \caption{Non-trivial bounds for $\hat{V}_0^\ast(\bar{x}_0)$ ($\bar{x}_0\in
   \mathbb{R}^{40}$); \emph{fm} and \emph{ps} is FTBU with \emph{fmincon}
   and \emph{patternsearch} respectively.}
   \label{tab:test_points}
\end{table}

\subsection{Conservativeness of FTBU}
\label{sub:consvFTBU}

Figure~\ref{fig:consv}(d) and (e) shows 
the relative error of FTBU with respect to DPBDA, with 
$\mathcal{T}={[-0.5,0.5]}^{2}$,
$\mathcal{S}={[-1,1]}^{2}$, grid spacing of $0.05$, and 
$\epsilon=0.01$.  
FTBU implemented using \emph{patternsearch} has
$77.57\%$ grid points with the relative error less than $30\%$ as compared to
$6.6\%$ grid points for \emph{fmincon}-based FTBU. This is also reflected in
Figure 2(a), (b), (c). 

The sharp rise in Figure 2(d) is due to points where
$\hat{W}_0^\ast(\bar{x}_0)=\epsilon$ and $\hat{V}_0^\ast(\bar{x}_0)>>\epsilon$,
resulting in a large relative error.  The conservativeness of FTBU
highlights the role of feedback %
in increasing the terminal time probability for any $\bar{x}_0\in \mathcal{X}$.
However, as seen in Table~\ref{tab:test_points}, for sufficiently large 
$\mathcal{S}, \mathcal{T}$, we obtain non-trivial lower bounds even for
high-dimensional systems.  Figure 2(f) and Table~\ref{tab:test_points} show
that FTBU with \emph{patternsearch} clearly outperforms \emph{fmincon} in %
the quality of the underapproximation, at the expense of computation time.

Lastly, note that the expected symmetry about the origin of the terminal time
probability for the system \eqref{eq:sys_I_chain} is not evident, unless a fine
grid is used (Figure 2(a)).
Table~\ref{tab:certificate} shows that FTBU can serve as a ``certificate''
for the validity of the grid spacing in DPBDA by relying on the conservativeness
established by Theorem~\ref{thm:consv}.  That is, the FTBU
underapproximation provides a grid-independent lower bound on the value function
$\hat{V}_0^\ast(\bar{x}_0)$ computed using DBPDA.
For example, for the double integrator, a grid spacing of $0.1$
will not give accurate results with DPBDA, since
$\hat{V}_0^\ast(\bar{x}_1)<\hat{W}_0^\ast(\bar{x}_1)$ contradicts 
Theorem~\ref{thm:consv}.%

\begin{table}
    \centering
    \tabcolsep=0.11cm
    \begin{tabular}{|c|c|c|c|c|}
    \hline									%
    Grid spacing & $0.1$ & $0.05$ & $0.01$ & $0.005$\\ \hline\hline
    $\hat{V}_0^\ast(\bar{x}_1)$ & $0.422$ & $0.506$ & $0.476$ & $0.478$ \\ \hline
    $\hat{V}_0^\ast(\bar{x}_2)$ & $0.527$ & $0.510$ & $0.483$ & $0.479$ \\ \hline
    Computation time (seconds) & $11.22$ & $42.68$ & $1206.59$ & $5710.06$ \\ \hline
    \end{tabular}
    \caption{Grid spacing in DPBDA ($n=2$, $\bar{x}_1=-\bar{x}_2={[0.1\
    0.9]}^\top$, $\mathcal{T}={[-0.5,0.5]}^{2}$, $\mathcal{S}={[-1,1]}^{2}$,
$\hat{W}_0^\ast(\bar{x}_1)=\hat{W}_0^\ast(\bar{x}_2)=0.436$, $ \epsilon =0.001$).}
    \label{tab:certificate}
\end{table}

\section{Conclusion}%
\label{sec:conc}

We show the conservativeness of the open-loop
formulation of the finite time horizon terminal hitting time stochastic
reach-avoid problem for stochastic linear systems, using conditional
expectations and sufficient conditions for Borel-measurability of
the value functions.  The open-loop formulation converts the verification
problem into a simpler optimization problem. The objective function is a
multi-dimensional integral, and an analytical expression of the integrand can be
obtained using Fourier transforms. For Gaussian disturbances, the objective
function can be evaluated efficiently and the optimization problem is
log-concave. Because the underapproximation technique does not rely on a grid, it
mitigates the curse of dimensionality, and provides non-trivial lower bounds on
the stochastic reach-avoid probability.  The method is demonstrated on 40D dynamical system. 

\bibliographystyle{IEEEtran}
\bibliography{IEEEabrv,shortIEEE,refs}
\end{document}